\documentclass[11pt,a4paper]{article}
\usepackage[numbers,sort&compress]{natbib}
\usepackage{amsmath,amssymb,amsthm}
\usepackage[margin=1in]{geometry}
\usepackage[english]{babel}
\usepackage{hyperref}
\usepackage{graphicx}
\usepackage{physics}
\usepackage{bm}
\usepackage{caption}
\usepackage{subcaption}
\usepackage{color}
\usepackage{fancybox}
\usepackage{amsmath}
\usepackage{amsfonts}
\usepackage{cancel}
\usepackage{mathrsfs}
\usepackage{multirow}
\usepackage{pstricks}
\usepackage{multicol}
\usepackage{bm}
\usepackage{tikz}
\usepackage{tikz-feynman}
\usepackage{feynmp-auto}
\usepackage{amssymb}
\usepackage{float}
\usepackage{xcolor}
\usepackage{xfrac}
\usepackage{chemfig}
\usepackage{eso-pic}
\usepackage[numbers]{natbib}
\usepackage[T1]{fontenc}
\usepackage{titlesec}
\usepackage{authblk}
\usepackage{enumitem}
\titleformat{\section}
{\bfseries}{\thesection.}{1em}{\normalfont\bfseries}
\usetikzlibrary{positioning, arrows.meta}
\titleformat{\subsection}
{\bfseries}{\thesubsection}{1em}{\small\bfseries}
\usetikzlibrary{positioning, arrows.meta}
\hypersetup{
	colorlinks=true,       
	linkcolor=blue,       
	citecolor=blue,     
	filecolor=blue,     
	urlcolor=blue          
}
\theoremstyle{definition}
\newtheorem{theorem}{Theorem}

\newtheorem{definition}[theorem]{Definition}
\newtheorem{example}[theorem]{Example}

\newtheorem{prop}[theorem]{Proposition}
\title{\Large \textbf{A Relation Between the Chrestenson Operator, Weyl Operator Basis, and Kronecker-Pauli Operator Basis}}
\author[1, 2]{Mickaya A. Razanaparany\footnote{Email: \href{mailto:mickaya@aims.ac.za}{\texttt{mickaya@aims.ac.za}}}}
\author[3, 4]{Christian Rakotonirina\footnote{Email: \href{mailto:rakotonirinachristianpierre@gmail.com}{\texttt{rakotonirinachristianpierre@gmail.com}}}}
\affil[1]{\small\textit{Physique des Hautes Energies, PHE, Université d'Antananarivo, Madagascar}}
\affil[2]{\small\textit{African Institute for Mathematical Sciences, AIMS, South Africa}}
\affil[3]{\small\textit{Institut Supérieur de Technologie d'Antananarivo, IST-T, Madagascar}}
\affil[4]{\small\textit{Laboratoire de la Dynamique de l'Atmosphère, du Climat et des Océans, DyACO, Université d'Antananarivo, Madagascar}}
\date{}
\begin{document}
\maketitle
\begin{abstract}
	Within the framework of quantum theory, we review the Chrestenson operator, the Weyl operator basis, and the Kronecker-Pauli operator basis in $d$-dimensional Hilbert spaces using Dirac notation, where $d$ is a prime integer strictly greater than 2. We establish a new algebraic relation connecting these operators and present the cases $d=3$ and $d=5$ as illustrative examples.
	
	\vspace*{0.5em}
\end{abstract}
\section{\textsc{Introduction}}
Quantum theory is formulated in terms of operators acting on a Hilbert space $\mathcal{H}$. All physical quantities of interest, such as observables, density matrices, and evolution operators, are represented by operators that can be expressed as linear combinations of a set of basic operators. As illustrative examples, linear operators map a vector space to itself, while specific classes are distinguished by their adjoint properties. Unitary operators satisfy the condition that their adjoint equals their inverse, and Hermitian operators are equal to their adjoint. The Pauli operators form a familiar case, where the identity $\mathbb{I}$ leaves a qubit invariant and the remaining operators:
\begin{equation}
	X=\begin{pmatrix}
		0&1\\
		1&0
	\end{pmatrix},\qquad
	Y=\begin{pmatrix}
		0&-i\\
		i&0
	\end{pmatrix},\qquad\text{ and }\qquad
	Z=\begin{pmatrix}
	1&0\\
	0&-1
	\end{pmatrix},
\end{equation}
generate elementary transformations~\citep{quantMec, unitOPER}. Understanding operators is essential for the foundations of quantum theory, the design of quantum algorithms, and the handling of subtle computations involving higher-order matrices~\cite{quantMec, unitOPER, quantum, quantLECT}.

In a $d$-dimensional Hilbert space $\mathcal{H}$, quantum states can be represented using Weyl operators. The Weyl operators are unitary, traceless, involve the $d$-th roots of unity, and form a group~\cite{weyl222}. While the $3 \times 3$ Kronecker--Pauli matrices (KPMs) have been explored in~\cite{kpmtrois}, and more general constructions have been systematically studied in~\citep{setOfKPMs, kpoperat}. The Kronecker-Pauli operators are unitary, hermitian, and have, for $d$ a prime integer, trace equal to unity. The relation between these operators has not yet been fully established. The aim of this paper is to introduce the Chrestenson operator and to show how it provides a relation between Weyl and Kronecker-Pauli operators in $d$-dimensional Hilbert spaces, where $d$ is a prime integer with $d>2$.

The paper is organized as follows. Section~\ref{defOPDEUX} reviews the relevant operator definitions, together with their properties and illustrative examples. Section~\ref{secTNTRREalt} derives expressions that relate the Chrestenson operator, the Weyl operator basis, and the Kronecker-Pauli operator basis for $d>2$ a prime integer, and presents their matrix relations for $d=3$ and $d=5$. Section~\ref{ccl} concludes the paper and provides some outlook.
\section{\textsc{Operator Definitions}}\label{defOPDEUX}
As is well known in quantum computation, the Hadamard transform is represented by
\begin{equation}
	H=\frac{1}{\sqrt{2}}
	\begin{pmatrix}
		1&1\\
		1&-1
	\end{pmatrix}.
\end{equation}
When a qubit in the span $\{|0\rangle, |1\rangle\}$ is acted upon the Hadamard gate, the resulting state has equal probability outcomes when measured~\cite{quantum}. Moving to a qutrit, $i.e.$ a state in the span $\{|0\rangle, |1\rangle,|2\rangle\}$, the so-called radix-3 Chrestenson transform is given by
\begin{equation}
	C_3=\frac{1}{\sqrt{3}}
	\begin{pmatrix}
		1&1&1\\
		1&w&w^2\\
		1&w^2&w
	\end{pmatrix}.
\end{equation}
In higher-dimensional quantum systems, the Chrestenson transform matrices can be seen as an extension of the Hadamard transform. More generally, for a qudit, $i.e.$ a state in the span $\{|0\rangle, |1\rangle,\,\dots, |d-1\rangle\}$, the radix-$d$ Chrestenson transform matrix is a form of the discrete Fourier transform (DFT), expressed as
\begin{equation}\label{radixdDFT}
	C_d=\frac{1}{\sqrt{d}}
	\begin{pmatrix}
		1&1&1&\cdots&1\\
		1&w&w^2&\cdots&w^{(d-1)}\\
		1&w^2&w^4&\cdots&w^{2(d-1)}\\
		\vdots&\vdots&\vdots&\ddots&\vdots\\
		1&w^{d-1}&w^{2(d-1)}&\cdots&w^{(d-1)^2}
	\end{pmatrix},
\end{equation}
where $w:=e^{2i\pi/d}$ is the $d$-th root of unity\footnote{The $d$-th root of unity are the solutions of $w^d=1$ in $\mathbb{C}$, implies that $1+w+\cdots+w^{d-1}=0$ if $w\neq 1$ for $d$ integer.}~\citep{logic, chresHIGH}.\\

Let us now define the Chrestenson operators, which is the matrix in eq.~\eqref{radixdDFT} expressed in terms of Dirac Bra and Ket notation.
\begin{definition}
	Let $\mathcal{H}$ be a Hilbert space of dimension $d\geq 2$, with computational basis $\{|0\rangle, |1\rangle,\,\dots\allowbreak, |d-1\rangle\}$ spanning the space, $\mathcal{H}=\texttt{span}\{|0\rangle, |1\rangle,\,\dots, |d-1\rangle\}$. The \textit{Chrestenson operators} $C_d$ is defined as
	\begin{equation}\label{chresOP}
		C_d=\frac{1}{\sqrt{d}}\sum_{x,y=0}^{d-1}w^{xy}\Big|y\Big\rangle \Big\langle x\Big|.
	\end{equation}
\end{definition}
\noindent The basic properties of the Chrestenson operators $C_d$, defined in eq.~\eqref{chresOP}, follow directly from its construction. In particular, $C_d$ is a unitary operator, satisfying $C_d^\dagger C_d = \mathbb{I}_d = C_d C_d^\dagger.$ As a consequence of its unitarity, the inverse of $C_d$ coincides with its adjoint, that is, $C_d^{-1} = C_d^\dagger$.\\

A key operator in our formulation is the Weyl operator basis, defined as follow~\citep{weyl1111, weyl222}:
\begin{definition}The \textit{Weyl operators} $U_{nm}$ is defined by
	\begin{equation}\label{weylOP}
		U_{nm}=\sum_{k=0}^{d-1}w^{kn}\Big|k\Big\rangle\Big\langle \big(k+m\big)\bmod d\Big|,
	\end{equation}
	where $n,m\in\{0,1,\,\dots, d-1\}$.
\end{definition}
\noindent These operators satisfy several properties: they form a set of $d^2$ matrices, consisting of the identity operator and $d^2-1$ traceless operators; All $d^2$ operators are unitary and constitute an orthonormal basis for the space of linear operators on $\mathcal{H}$ with respect to the Hilbert-Schmidt inner product, that is
\begin{equation}
	\mathrm{Tr}\!\left(U_{nm}^\dagger U_{pq}\right)
	= d\,\delta_{np} \delta_{mq}.
\end{equation}

\begin{example}
	Let us take as an example the case $d=3$ which working in the computation basis:
	\begin{equation}
		|0\rangle = \begin{pmatrix}
			1\\0\\0
		\end{pmatrix},\qquad |1\rangle = \begin{pmatrix}
		0\\1\\0
	\end{pmatrix}, \qquad |2\rangle = \begin{pmatrix}
	0\\0\\1
	\end{pmatrix}.
	\end{equation}
	The 9 matrices are the following:
	\begin{equation}
		\begin{array}{l l l}
				U_{00} = \begin{pmatrix}
					1&0&0\\0&1&0\\0&0&1
				\end{pmatrix},
			&\quad
				U_{01} = \begin{pmatrix}
					0&1&0\\0&0&1\\1&0&0
				\end{pmatrix},
			&\quad
				U_{02} = \begin{pmatrix}
					0&0&1\\1&0&0\\0&1&0
				\end{pmatrix},
			\\[2.5em]
				U_{10} = \begin{pmatrix}
					1&0&0\\0&w&0\\0&0&w^2
				\end{pmatrix},
			&\quad
				U_{11} = \begin{pmatrix}
					0&1&0\\0&0&w\\w^2&0&0
				\end{pmatrix},
			&\quad
				U_{12} = \begin{pmatrix}
					0&0&1\\w&0&0\\0&w^2&0
				\end{pmatrix},
			\\[2.5em]
				U_{20} = \begin{pmatrix}
					1&0&0\\0&w^2&0\\0&0&w
				\end{pmatrix},
			&\quad
				U_{21} = \begin{pmatrix}
					0&1&0\\0&0&w^2\\w&0&0
				\end{pmatrix},
			&\quad
				U_{22} = \begin{pmatrix}
					0&0&1\\w^2&0&0\\0&w&0
				\end{pmatrix}.
		\end{array}
	\end{equation}
\end{example}
We introduce the following definitions, which play a central role in the analysis of Kronecker-Pauli matrices and operators~\citep{kpmtrois, setOfKPMs, kpoperat}.
\begin{definition}
	Let $d>2$ be an integer, a family of $d\times d$-KPMs is a collection $\big\{\Pi_k\big\}_{k=0}^{d^2-1}$ such that the following properties being satisfied:
	\begin{enumerate}[label=(\arabic*)]
		\item $S_{d\otimes d}=\frac{1}{d}\sum_{k=0}^{d^2-1}\Pi_k\otimes \Pi_k$ \quad ($d\otimes d$ swap operator)
		\item $\Pi_k^\dagger=\Pi_k, \text{ for } 0\leq k \leq d^2-1$\quad (hermiticity)
		\item $\Pi_k^2=\mathbb{I}_d, \text{ for } 0\leq k \leq d^2-1$\quad (square root of the unit)
		\item $\Tr\big(\Pi_k^\dagger\Pi_\ell\big)=d\delta_{k\ell}, \text{ for }0\leq k \text{ and } \ell\leq d^2-1$\quad (orthogonality)
	\end{enumerate}
	where $\delta_{k\ell}$ is the Kronecker symbol\footnote{The Kronecker delta $\delta_{k\ell}$ equals $1$ if $k=\ell$ and $0$ otherwise.}.
\end{definition}
For a prime number $d$, Kronecker-Pauli operators are defined so that their matrices in the standard basis form a complete set of $d\times d$-KPMs~\cite{kpoperat}.
\begin{definition}
	For a prime integer $d>2$, we define a \textit{Kronecker-Pauli operators} $\Pi_{nm}$ as
	\begin{equation}
		{\prod}_{nm}=\sum_{k=0}^{d-1}w^{(k-n)m}\Big|k \Big\rangle\Big\langle (-k+2n)\bmod d\Big|,
	\end{equation}
	where $n,m\in\{0,1,\,\dots, d-1$\}.\\
\end{definition}
\begin{example}\label{exampleSIX}
	For $d=3$, the $\tau_k$ with $k=1,2,\,\dots, 9$ KPMs are the following:
	\begin{equation}
		\begin{array}{l l l}
			\tau_1= \begin{pmatrix}
				1&0&0\\0&0&1\\0&1&0
			\end{pmatrix},
			&\quad
			\tau_2 = \begin{pmatrix}
				1&0&0\\0&0&w\\0&w^2&0
			\end{pmatrix},
			&\quad
			\tau_3 = \begin{pmatrix}
				1&0&0\\0&0&w^2\\0&w&0
			\end{pmatrix},
			\\[2.5em]
			\tau_4 = \begin{pmatrix}
				0&0&1\\0&1&0\\1&0&0
			\end{pmatrix},
			&\quad
			\tau_5 = \begin{pmatrix}
				0&0&w\\0&1&0\\w^2&0&0
			\end{pmatrix},
			&\quad
			\tau_6 = \begin{pmatrix}
				0&0&w^2\\0&1&0\\w&0&0
			\end{pmatrix},
			\\[2.5em]
			\tau_7 = \begin{pmatrix}
				0&1&0\\1&0&0\\0&0&1
			\end{pmatrix},
			&\quad
			\tau_8 = \begin{pmatrix}
				0&w&0\\w^2&0&0\\0&0&1
			\end{pmatrix},
			&\quad
			\tau_9 = \begin{pmatrix}
				0&w^2&0\\w&0&0\\0&0&1
			\end{pmatrix}.
		\end{array}
	\end{equation}
\end{example}
\section{\textsc{Chrestenson, Weyl, and Kronecker-Pauli Relationship}}\label{secTNTRREalt}
The goal is to explore the relation between the Weyl operator basis and the Kronecker-Pauli operator basis. To do so, we are inspired by the interaction of the Chrestenson and $X, Y, Z$ for ternary quantum computing~\cite{TER3}.
\begin{prop}
	Let $d>2$ be a prime integer, let $C_d$ denote the Chrestenson transform, $U_{nm}$ the Weyl operators, and $\Pi_\ell$ the Kronecker-Pauli operators. Then, for every $n,m\in \{0,1,\,\dots, d-1\}$ there exist integers $k\in \{0,1,\,\dots, d-1\}$ and $\ell\in \{1,2,\,\dots, d^2\}$ such that
	\begin{equation}\label{propEQ10}
		\large C_dU_{nm}C_d=w^k\Pi_\ell,
	\end{equation}
	where $w$ is the $d$-th root of unity.
\end{prop}
\begin{proof}
	We proceed to compute the product directly, as follows:
	\begin{align}
		C_dU_{nm}C_d&=\frac{1}{d}\sum_{x_1,y_1=0}^{d-1}\sum_{k=0}^{d-1}\sum_{x_2,y_2=0}^{d-1}w^{x_1y_1}\Big|y_1\Big\rangle \Big\langle x_1\Big|w^{kn}\Big|k\Big\rangle\Big\langle \big(k+m\big)\bmod d\Big|w^{x_2y_2}\Big|y_2\Big\rangle \Big\langle x_2\Big|\nonumber\\
		&=\frac{1}{d}\sum_{x_1,y_1=0}^{d-1}\sum_{k=0}^{d-1}\sum_{x_2,y_2=0}^{d-1}w^{x_1y_1}w^{kn}w^{x_2y_2}\Big|y_1\Big\rangle \delta_{x_1,k}\delta_{(k+m)\bmod d,y_2}\Big\langle x_2\Big|.\label{catf5}
	\end{align}
	It is known that:
	\begin{equation}\label{eisntSUM}
		\begin{aligned}
			&\sum_{x_1=0}^{d-1}\big(w^{y_1}\big)^{x_1}\delta_{x_1,k}=w^{y_1k},\\
			&\sum_{y_2=0}^{d-1}\big(w^{x_2}\big)^{y_2}\delta_{(k+m)\bmod d,y_2}=w^{x_2\big[(k+m)\bmod d\big]}.
		\end{aligned}
	\end{equation}
	By substituting eq.~\eqref{eisntSUM} into eq.~\eqref{catf5}, we find
	\begin{equation}\label{befMOD}
		C_dU_{nm}C_d=\frac{1}{d}\sum_{y_1=0}^{d-1}\sum_{k=0}^{d-1}\sum_{x_2=0}^{d-1}w^{y_1k}w^{kn}w^{x_2\big[(k+m)\bmod d\big]}\Big|y_1\Big\rangle\Big\langle x_2\Big|.
	\end{equation}
	Now, let $w$ be a $d$-th root of unity. For any integers $x_2,k,m$, we have $(k+m)\bmod d=k+m-qd$ for some $q\in\mathbb{Z}$. Then, using $w^d=1$, it follows that
	\begin{equation*}
		w^{x_2\big[(k+m)\bmod d\big]}=w^{x_2(k+m-qd)}=w^{x_2(k+m)}w^{-x_2qd}=w^{x_2(k+m)}.
	\end{equation*}
	Under these considerations, eq.~\eqref{befMOD} can be rewritten as
	\begin{equation}\label{befGEOseries}
		C_dU_{nm}C_d=\frac{1}{d}\sum_{y_1=0}^{d-1}\sum_{k=0}^{d-1}\sum_{x_2=0}^{d-1}w^{k(y_1+n+x_2)}w^{x_2m}\Big|y_1\Big\rangle\Big\langle x_2\Big|.
	\end{equation}
	As we see, the sum over $k$ forms a geometric series. Let us check it carefully as shown below:
	\begin{equation}
		\begin{aligned}
			\sum_{k=0}^{d-1}w^{k(y_1+n+x_2)}&=\sum_{k=0}^{d-1}\left(w^{y_1+n+x_2}\right)^k\\
			&=\begin{cases}
				d &\text{ if }y_1+n+x_2\bmod d=0\\
				0 &\text{ otherwise}
			\end{cases}\\
			&=d\delta_{(y_1+n+x_2)\bmod d,0}.
		\end{aligned}
	\end{equation}
	Eq.~\eqref{befGEOseries} become
	\begin{equation}
		C_dU_{nm}C_d=\sum_{y_1=0}^{d-1}\sum_{x_2=0}^{d-1}\delta_{(y_1+n+x_2)\bmod d,0}\text{ }w^{x_2m}\Big|y_1\Big\rangle\Big\langle x_2\Big|.
	\end{equation}
	Since we are interested in the Kronecker delta taking the value 1, we must have $(y_1+n+x_2)\bmod d=0$. Then, we obtain 
	\begin{equation}\label{myEXPCUC}
		C_dU_{nm}C_d=\sum_{x_2=0}^{d-1}\text{ }w^{x_2m}\Big|\big(-x_2-n\big)\bmod d\Big\rangle\Big\langle x_2\Big|.
	\end{equation}
	Let us consider the expression $(-x_2-n)\bmod d$ and distinguish two cases according to the parity of $n$:
	\begin{itemize}
		\item[$\circ$] For $n$ \textit{odd}, adding $d$ yields $(-x_2-n)\bmod d=\big[-x_2+(d-n)\big]\bmod d$. Since $d>2$ is assumed to be a prime integer, $d$ is odd. Consequently, $d-n$ is even whenever $n$ is odd, meaning that we can write $d-n=2h$ with $h=\frac{1}{2}(d-n)$. Then
		\begin{equation}\label{formGENPK}
			C_dU_{nm}C_d=\sum_{x_2=0}^{d-1}\text{ }w^{x_2m}\Big|\big(-x_2+2h\big)\bmod d\Big\rangle\Big\langle x_2\Big|.
		\end{equation}
		\item[$\circ$] For $n$ \textit{even} and $n>1$, writing $n=2p$ and adding $2d$ gives $(-x_2-n)\bmod d=\big[-x_2+2(d-p)\big]\bmod d$. Defining $h=d-p$, this expression again takes the form of eq.~\eqref{formGENPK}.
	\end{itemize}

	\noindent In both cases, the transformed operator has the same structural form: a phase factor multiplied by a permutation of the computational basis. Consequently, we obtain the general relation
	\begin{equation}
		\large C_dU_{nm}C_d=w^{k}\Pi_{\ell},
	\end{equation}
	where $n,m,k\in \{0,1,\,\dots, d-1\}$ and $\ell\in\{1,\,\dots, d^2\}$. This result is in complete agreement with eq.~\eqref{propEQ10}, which completes the proof of the proposition.
\end{proof}

\begin{example}For $d=3$:
	\begin{equation*}
		\begin{array}{l l l}
			C_3U_{00}C_3=\tau_1&\qquad\qquad C_3U_{01}C_3=\tau_3 &\qquad\qquad C_3U_{02}C_3=\tau_2\\
			C_3U_{10}C_3=\tau_4&\qquad\qquad C_3U_{11}C_3=w\tau_5 &\qquad\qquad C_3U_{12}C_3=w^2\tau_6\\
			C_3U_{20}C_3=\tau_7&\qquad\qquad C_3U_{21}C_3=w^2\tau_9 &\qquad\qquad C_3U_{22}C_3=w\tau_8.
		\end{array}
	\end{equation*}
	The third root of unity is denoted by $w$, and the operators $\tau_i$, for $i=1,2,\,\dots,9$, correspond to the KPMs presented in example~\ref{exampleSIX}.
\end{example}
\begin{example} For $d=5$:\\
	\noindent
	To avoid confusion, we denote the fifth root of unity by $\eta=e^{\frac{2i\pi}{5}}$, and $\chi_j$ for $j=1,2,\,\dots,25$ are the KPMs from~\cite{setOfKPMs}.
	\begin{equation*}
		\begin{array}{l l l l l}
			C_5U_{00}C_5=\chi_{1}&C_5U_{01}C_5=\chi_{2}&C_5U_{02}C_5=\chi_{3}&C_5U_{03}C_5=\chi_{4}&C_5U_{04}C_5=\chi_{5}\\
			C_5U_{10}C_5=\chi_{11}&C_5U_{11}C_5=\eta^2\chi_{12}&C_5U_{12}C_5=\eta^4\chi_{13}&C_5U_{13}C_5=\eta\chi_{14}&C_5U_{14}C_5=\eta^3\chi_{15}\\
			C_5U_{20}C_5=\chi_{21}&C_5U_{21}C_5=\eta^4\chi_{22}&C_5U_{22}C_5=\eta^3\chi_{23}&C_5U_{23}C_5=\eta^2\chi_{24}&C_5U_{24}C_5=\eta\chi_{25}\\
			C_5U_{30}C_5=\chi_{6}&C_5U_{31}C_5=\eta\chi_{7}&C_5U_{32}C_5=\eta^2\chi_{8}&C_5U_{33}C_5=\eta^3\chi_{9}&C_5U_{34}C_5=\eta^4\chi_{10}\\
			C_5U_{40}C_5=\chi_{16}&C_5U_{41}C_5=\eta^3\chi_{17}&C_5U_{42}C_5=\eta\chi_{18}&C_5U_{43}C_5=\eta^4\chi_{19}&C_5U_{44}C_5=\eta^2\chi_{20}.
		\end{array}
	\end{equation*}
\end{example}

\section{\textsc{Conclusion and Outlook}}\label{ccl}
For $d$-dimensional system, with $d$ a prime integer greater than 2, our results demonstrate that the Chrestenson operators serves as a connecting transformation, mapping Weyl operators to Kronecker-Pauli operators up to phase factors and index permutations. This framework provides a systematic and transparent interpretation of these operator bases within a single formalism, highlighting their distinct roles in operator decompositions and their interrelations. From the perspective of quantum computation, this relation suggests the existence of equivalences between certain quantum gates or circuits when expressed in different operator bases. These equivalences could be exploited to simplify circuit designs or to translate quantum algorithms between alternative representations.

For future research, it would be of considerable interest to inverstigate the effects of taking the conjugate of a Chrestenson operator, i.e., $C_d^\dagger U_{nm}C_d$ and/or $C_dU_{nm}C_d^\dagger$. Extending this formulation to all finite-dimensional Hilbert spaces would also be a promising direction for further study. Additionally, applying this relation to aspects of ternary reversible computing, as developed in~\cite{TER3}, opens promising avenues for qutrit-level quantum computation. Furthermore, revisiting the sets of matrices used as models for error-correction schemes in ternary and higher-dimensional quantum systems is worthwhile~\citep{qecTernary, qecHAFA}, since these are KPMs and our relation show that they can be expressed entirely in terms of Chrestenson operators and Weyl operators.

\section*{Author Contributions}
All authors contributed equally to the conception, proof of idea, computation, and writing of this work.

\section*{Declarations}
\begin{itemize}
	\item[\textbf{-}] \textbf{Funding:} This research received no funding.
	\item[\textbf{-}] \textbf{Conflicts of interest:} The authors declare no conflicts of interest.
	\item[\textbf{-}] \textbf{Data availability:} No external data were used in this research.
\end{itemize}

\end{document}